\documentclass[aps,pra,twocolumn,superscriptaddress,longbibliography]{revtex4-2} 
\usepackage{subfigure,subcaption,caption}

\usepackage{graphicx,color,appendix,physics,wasysym,xcolor,xparse}
\usepackage{appendix} 
\usepackage{amsmath,amssymb,amsfonts,tabularx}
\usepackage{bm,bbm,euscript,braket,setspace}
\usepackage{hyperref}
\hypersetup{colorlinks,linkcolor={red},citecolor={blue},urlcolor={blue}}

\usepackage{amsthm}
\theoremstyle{plain}

\newtheorem{lemma}{Lemma}

\usepackage{verbatim,multirow}

\usepackage[utf8]{inputenc}
\usepackage[OT1]{fontenc}
\usepackage{lmodern}
\usepackage{soul}
\usepackage[normalem]{ulem}

\newcommand{\cB}{\mathcal{B}}
\newcommand{\cC}{\mathcal{C}}

\newcommand{\cE}{\mathcal{E}}

\newcommand{\cH}{\mathcal{H}}

\newcommand{\cO}{\mathcal{O}}

\newcommand{\cR}{\mathcal{R}}

\newtheorem{theorem}{Theorem}

\begin{document}
\title{Noise-adapted Quantum Error Correction for Non-Markovian Noise}
\author{Debjyoti Biswas}
\email{deb7059@gmail.com}
\affiliation{Department of Physics, Indian Institute of Technology Madras, Chennai 600036, India.}
\affiliation{Center for Quantum Information, Communication and Computing, Indian Institute of Technology Madras, Chennai 600036, India.}
\author{Shrikant Utagi}
\email{shrikant.phys@gmail.com}
\affiliation{Department of Physics, Indian Institute of Technology Madras, Chennai 600036, India.}
\author{Prabha Mandayam}
\email{prabhamd@physics.iitm.ac.in}
\affiliation{Department of Physics, Indian Institute of Technology Madras, Chennai 600036, India.}
\affiliation{Center for Quantum Information, Communication and Computing, Indian Institute of Technology Madras, Chennai 600036, India.}

\begin{abstract}
We consider the problem of quantum error correction (QEC) for non-Markovian noise. Using the well known Petz recovery map, we first show that conditions for approximate QEC can be easily generalized for the case of non-Markovian noise, in the strong coupling regime where the noise map becomes non-completely-positive at intermediate times. While certain approximate QEC schemes are ineffective against quantum non-Markovian noise, in the sense that the fidelity vanishes in finite time, the Petz map adapted to non-Markovian noise uniquely safeguards the code space even at the maximum noise limit. Focusing on the case of non-Markovian amplitude damping noise, we further show that the non-Markovian Petz map also outperforms the standard, stabilizer-based QEC code. Since implementing such a non-Markovian map poses practical challenges, we also construct a Markovian Petz map that achieves similar performance, with only a slight compromise on the fidelity. 
 
\end{abstract}

\maketitle

\section{Introduction}
When a quantum system undergoes decoherence \cite{breuer2002theory,weiss2008quantum}, it can do so in either of the ways. Under certain conditions, the system evolves according to Markovian dynamics in which it unidirectionally loses its information content to the environment. More generally, systems undergo non-Markovian dynamics originating due to strong system-environment coupling. This in turn gives rise to the so-called \textit{information backflow} due to which the lost information comes back, at least partly, to the system \cite{breuer2009measure,RHP10}. The latter is generally termed the ``quantum memory effect" or \emph{non-Markovianity}, for short. In fact, information backflow represents the strongest form of non-Markovianity, while there exist processes that capture weaker forms of non-Markovian dynamics, see for example, \cite{RHP14,breuer2016colloquium,li2018concepts,shrikant2023NM}.

Quantum noise is generally detrimental to quantum information processing tasks and the current quantum hardware relies on various methods of mitigating or suppressing the noise \cite{cai2023quantum}. One of the ways to protect quantum information or reverse the effects of noise in quantum computers is to implement quantum error correction (QEC) \cite{nielsen_chuang_2010} -- an essential strategy for scaling up noisy quantum hardware to robust, fault-tolerant quantum computers. Much of the literature so far considers the noise to be Markovian for the purposes of developing QEC techniques which include the general purpose (or, standard) QEC \cite{nielsen_chuang_2010,terhal_qec} or approximate QEC (AQEC) \cite{barnum2002,leung,beny2010general,ng2010simple,mandayam2012towards}. 

More recently, there have been attempts to address the question of non-Markovian errors \cite{oreshkov2007continuous,shabani2009maps,len2018open,byrd2020sufficient,cao2023quantum}, which occur due to either drift \cite{proctor2020detecting}, cross-talk \cite{white2023filtering}, or, strong coupling with the environment \cite{RHP14,breuer2016colloquium,vega2017dynamics,li2018concepts,shrikant2023NM}. Moreover, non-Markovian memory effects might pose a major challenge to fault-tolerant quantum computation \cite{terhal_qec,terhal2005fault,aharonov2006fault,lemberger2017effect}, motivating the need for a more general theory of non-Markovian QEC \cite{pollock2018non,tanggara2024strategic}.

In this work, we study the efficacy of standard QEC and approximate, noise-adapted QEC in the presence of non-Markovian noise. We consider the approach of channel-adapted QEC based on the Petz recovery map \cite{barnum2002, ng2010simple, mandayam2012towards}, wherein  the recovery operation exactly adapts to the error (or, Kraus) operators. We first show that this approach leads to a generalization of the known approximate QEC conditions for the case of not completely positive (NCP) noise maps. This then motivates us to study the efficacy of the Petz map in correcting for non-Markovian errors.

We then move on to the specific example of non-Markovian amplitude damping noise,  which is an important  physical example of an NCP map~\cite{garraway1997nonperturbative,breuer2016colloquium}. In this case, we show that the non-Markovian Petz map which exactly adapts to the noise performs better than other channel-adapted QEC strategies~\cite{leung} and even outperforms the $5$-qubit stabilizer code. However, such a non-Markovian recovery channel may not be amenable for a physical implementation via a quantum circuit. We therefore construct a Markovian Petz channel that adapts to the structure of the noise operators, but not to the noise strength and find that the  Markovian Petz recovery channel performs nearly as well as the non-Markovian Petz map. However, there a price one pays for when the Petz channel is not exactly adapted to the noise strength: the composite QEC channel becomes non-unital, which in turn causes certain codes to be sub-optimal in terms of the error correction fidelity. We also analyse the non-Markovian nature of the composite QEC superchannel that ultimately explains the behavior of the worst-case fidelity for different QEC protocols.

The rest of the paper is structured as follows. Sec.~\ref{sec:prelim} is dedicated to introducing the concept of quantum non-Markovianity of quantum channels and providing a brief review on the literature of non-Markovian quantum error correction as well as the  Petz map. In Sec.~\ref{sec:bounds_on_exact}, we derive bounds on the fidelity obtained using perfectly correctable codes for non-Markovian noise. In Sec.~\ref{sec:noise-adapted}, we derive similar bounds for approximate quantum codes and noise-adapted recovery maps. This leads to approximate QEC conditions that are sufficient for a quantum code to correct for non-Markovian errors using the Petz recovery channel map. Finally, in Sec.~\ref{sec:examples} we consider the case of non-Markovian amplitude damping noise, and show that a Petz recovery map adapted to the noise in the Markovian regime suffices for approximate QEC, albeit at the cost of making the QEC superchannel non-unital, as discussed in Sec. \ref{sec:nonunital-consequence}. We summarize our results and outline future directions in Sec.~\ref{sec:concl}.

\section{Preliminaries \label{sec:prelim}}
\subsection{Non-Markovianity of quantum dynamical maps}
We begin with a brief discussion on the physical origins for non-Markovianity in the evolution of open quantum systems. There are a number of approximations that go into formulating a master equation that describes open system evolution, namely, (i) the initial factorization assumption, where the joint system-environment (S-E) state is assumed to be of the  product form (ii) the Born-Markovian approximation, in which a large separation between system and environment time-scale is assumed and (iii) the rotating-wave approximation (RWA), in which fast (or, counter)-rotating terms in the interaction Hamiltonian are ignored. Relaxing (i) or (iii) leads to the system evolution map being non-completely-positive (NCP) \cite{li2018concepts}. 

However, retaining (i) and (iii) but relaxing (ii) leads to a time-dependent master equation for the system density operator $\rho(t)$, that takes the form of a time-dependent GKSL-like equation (in the canonical form) \cite{hall2014canonical}.
\begin{align}
\frac{d \rho(t)}{dt} &= \mathcal{L}(t)[\rho(t)] \nonumber \\ &= \sum_j \Gamma_j(t) \bigg(L_j(t) \rho(t) L_j^\dagger(t) - \frac{1}{2}\{L_j^\dagger(t) L_j(t), \rho(t)\} \bigg).
\label{eq:gksl-like}
\end{align}
where $\{L_{j}(t)\}$ are called the jump operators,  $\{\Gamma_j(t)\}$ are the canonical decay rates, both of which maybe time-dependent in general, and $\{A,B\} = AB + BA$ is the anti-commutator for any two operators $A$ and $B$. $\mathcal{L}(t)$ denotes the time-dependent Lindbladian operator associated with the open system dynamics. The process in Eq. (\ref{eq:gksl-like}) is termed time-\textit{in}homogeneous Markovian when all the canonical decay rates $\Gamma_j(t)$ are positive for all times, in which case all the approximations from (i)-(iii) are assumed to be in place, implying that the environment operators are \textit{not} delta-correlated in time~\cite{breuer2016colloquium,li2018concepts}.

The solution to the master equation is given by a dynamical map of the form 
\begin{align}
    \mathcal{E}(t,t_0) = \mathcal{T} \exp \left\{ \int^{t}_{t_0}  \mathcal{L}(\tau) d\tau \right\},
\end{align}
where $\mathcal{T}$ is the time-ordering operator and $\tau$ is intermediate time. When (ii) is lifted, the map may be non-Markovian with at least one of the canonical decay rates $\Gamma(t)$ becoming negative at certain interval of time-evolution \cite{garraway1997nonperturbative,breuer2002theory,RHP10,hall2014canonical}. The quantum map $\mathcal{E}$ corresponding to a given Lindbladian $\mathcal{L}$ is also often referred to as a quantum noise channel. 

Given a space of bounded operators $\rho \in \cB(\cH)$, a linear map  $\mathcal{E}: \cB(\cH) \rightarrow \cB(\cH)$ is said to be completely positive (CP) if and only if its corresponding \textit{Choi} matrix $\chi$ is positive:
\begin{align}
    \chi(t_2,t_0) = \left(\mathcal{E}(t_2,t_0)\otimes I \right)[\ketbra{\Psi}] \ge 0, \label{eq:CJmatrix}
\end{align}
where $\ket{\Psi} = \ket{00} + \ket{11}$ is an unnormalized Bell state. When this is guaranteed, the map can be written in the Kraus form as
\begin{align}
    \rho(t) = \mathcal{E}(t)[\rho(0)]  = \sum_j E_j(t) \rho(0) E^\dagger _j (t),
    \label{eq:kraus_form}
\end{align}
where $E_j(t)$ are the Kraus operators, and the initial time $t_0=0$ for simplicity. {Throughout the paper, time-dependence of Kraus operator is implicit and we drop the parameter in parenthesis henceforth.}

One of the ways to look at non-Markovian dynamics is via the notion of \textit{divisibility} of the quantum dynamical maps corresponding to the noise \cite{RHP10}. A quantum map $\mathcal{E}(t_2,t_0)$, taking a density matrix from $t_0$ to $t_2$, is divisible if it can be concatenated as,
\begin{equation}
    \mathcal{E}(t_2,t_0) = \mathcal{E}(t_2,t_1)\mathcal{E}(t_1,t_0)
    \label{eq:divisibility-condition}
\end{equation}
for all $t_2 \ge t_1 \ge t_0$. This equality breaks down for non-Markovian noise in the sense that the intermediate map $\mathcal{E}(t_2,t_1) = \mathcal{E}(t_2,t_0)\mathcal{E}^{-1}(t_1,t_0)$ is NCP. A prototypical example of such a channel is non-Markovian amplitude damping discussed in Appendix \ref{sec:NMAD}. For this noise channel, the full map $\mathcal{E}(t_2,t_0)$ from $t_0$ to $t_2$ is CPTP, but the intermediate map $\mathcal{E}(t_2,t_1)$ is not. For the noise model we consider in this work, even though the intermediate map $\mathcal{E}(t_2,t_1)$ is NCP, it is still linear, Hermiticity-preserving and trace-preserving (HPTP). 

It is important to mention that there exists a hierarchy of divisibility conditions for CPTP maps \cite{chruscinski2014degree}. Two main conditions relevant for non-Markovian dynamics are CP-indivisibility and P-indivisibility. A CPTP map is said to be P-\textit{in}divisible if the intermediate map $\mathcal{E}(t_2,t_1)$ is not even positive. In this sense, a P-indivisible noise is necessarily CP-indivisible but the converse may not be true \cite{chruscinski2018divisibility}. However, for a channel with single jump operator $L$ in Eq.~(\ref{eq:gksl-like}), for a system of dimension two, it is known that CP-divisibility and P-divisibility coincide \cite{breuer2016colloquium,chruscinski2014degree,chakraborty2019information}, which suffices for us in the present work. This concept becomes relevant to us in Section~\ref{sec:nm_pdiv}. 

\subsection{Quantum Error Correction and the Petz map \label{sec:petzmap}}
Quantum error correction (QEC) proceeds by encoding the information into a subspace of a larger Hilbert space. A QEC protocol is characterized by the codespace $\mathcal{C}$ into which the information is encoded, as well as the recovery operation $\mathcal{R}$ which is in general a CPTP map. 

A codespace $\mathcal{C}$ is said to be \emph{ perfectly correctable} for a noise channel $\mathcal{E}$ with Kraus operators $\{E_{i}\}$ if there exists a recovery map $\mathcal{R}$ such that $(\mathcal{R}\circ\mathcal{E})[\rho]\propto \rho$, for all states $\rho$ on the codespace. The well known Knill-Laflamme (KL) conditions provide a set of necessary and sufficient conditions for a noise channel to be perfectly correctable on a codespace $\mathcal{C}$, as follows~\cite{knill_laflamme}.
\begin{align}
    PE_i^{\dagger}E_jP &= \alpha_{ij}P, \label{eq:perfect-conditions}
\end{align}
where $P$ is the projector onto the codespace $\cC$ and $\alpha_{ij}$ are complex scalars that form the elements of a Hermitian matrix. 

On the other hand, when there exists a recovery map $\cR$ such that $\mathcal{R}\circ\mathcal{E}[\rho]$ is close in fidelity to the original state $\rho$, for all states in a codespace $\cC$, the noise channel $\cE$ is said to be \emph{approximately correctable} on $\cC$~\cite{leung}. The search for good approximate quantum codes naturally leads to the idea of \emph{noise-adapted} QEC, where the code and recovery map are tailored to a specific noise channel, in contrast to standard quantum codes that are designed to perfectly correct for Pauli errors. Noise-adapted QEC serves as an interesting route to resource-efficient protocols, when a specific noise is known to be dominant in a given physical setting.

Several variants noise-adapted recovery maps have been proposed in the literature, based on analytical~\cite{leung, ng2010simple} and numerical~\cite{fletcher2008} constructions. Here, we will primarily focus on a specific form of the recovery channel known as the Petz map~\cite{petz1986sufficient, barnum2002}. Corresponding to a quantum code $\cC$ and noise map $\cE$, the Petz recovery channel is defined as, 
\begin{align}
\cR_{P, \cE} [.] &= P\,\cE^{\dagger} \,\left[\cE[P]^{-1/2}(.)\cE[P]^{-1/2}\right]\,P \nonumber \\
&= \sum_{i=1}^{N} P  E_{i}^{\dagger}\left(\cE[P]^{-1/2}(.)\cE[P]^{-1/2} \right) E_{i} P, 
\label{eqn:code-specific-petz}
\end{align}
where $P$ is the projector onto the codespace $\cC$ and $\cE[P] = \sum_{i}E_{i}PE_{i}^{\dagger}$ denotes the action of the noise map on the codespace. This map is known to provide a universal and near-optimal strategy for any code used and noise channel~\cite{ng2010simple,mandayam2012}. The near-optimality of the Petz recovery channel is stated with respect to the worst-case fidelity on the codespace, defined as,
\begin{align}\label{eq:wcf}
    F^2_{\rm min} &= \underset{|\psi\rangle \in \cC}{\rm min} \,\,F^2(\cR\circ \cE [\,|\psi\rangle\langle\psi|\,],|\psi\rangle)
\end{align}
where $\,F^2(\rho,|\psi\rangle) = \langle\psi|\rho|\psi\rangle$ is the fidelity between the states $\rho$ and $|\psi\rangle$.

\subsection{QEC for HPTP noise maps}
We now briefly review some of the recent studies on QEC for non-Markovian errors. In a Markovian process, the noise strength reaches maximum only in the asymptotic time limit, ensuring that we are within the assumptions of standard quantum error correction~\cite{nielsen}. Conversely, in a non-Markovian process, the noise strength can fluctuate significantly, from zero to its maximum value, within a finite time domain, causing the majority voting scheme, and subsequently the error correction protocol, to fail. In Sec.~\ref{sec:examples} we show that this is indeed the case, using the example of non-Markovian amplitude damping.

In this work, we specifically focus on HPTP noise channels, for which the map could NCP at intermediate times. We first note that for any HPTP map, there exists an operator sum-difference representation \cite{shabani2009maps,omkar2015operator,byrd2020sufficient,cao2023quantum} of the form,
\begin{align}
    \cE^{ \rm HPTP}_{\tau \rightarrow t}[\rho] &= \sum_i {\rm sign}(i) E_i\rho E^\dagger_i,
    \label{eq:hptp-kraus}
\end{align}
for $t \geq \tau$, where the index $i$ corresponds to $i^{\rm th}$ eigenvalue of the Choi matrix (Eq.~(\ref{eq:CJmatrix})) of the corresponding dynamical map. {Once again, we note that all operators are time-dependent unless otherwise stated.} 

A simple example {of such a map} would be the inverse of a phase-flip channel, which is an NCP-HPTP map \cite{shabani2009maps}. In fact, any linear NCP map can be represented as a difference of two CP maps $\mathcal{E}_1$ and $\mathcal{E}_2$ \cite{yu2000positive}. For example, one can find maps $\cE_{1}$ and $\cE_{2}$ that lead to an operator sum-difference decomposition for the non-Markovian amplitude damping channel~\cite{omkar2015operator}. 

It was recently argued~\cite{cao2023quantum} that the original Knill-Laflamme conditions given in Eq.~(\ref{eq:perfect-conditions}) are sufficient for non-Markovian errors characterized by Kraus operators $\{\text{sign}(i),E_i(t)\}$. This argument implicitly assumes that the matrix of coefficients $\alpha_{ij}$ in Eq.~\eqref{eq:perfect-conditions} is always diagonalizable.  However, it has been shown earlier ~\cite{byrd2020sufficient} that the diagonal form of the KL conditions $P F^\dagger_i F_j P = d_{ii}\delta_{ij}P$ would not be satisfied if $\cE_2[P \rho P] \ne 0$, that is, the code space is mapped outside the positivity domain of the noise map, thus causing the errors to be indistinguishable. Moreover, the $\alpha_{ij}$ matrix becomes pseudo-Hermitian for NCP noise and there is no guarantee that there exists a (pseudo-)unitary operator that diagonalizes it \cite{byrd2022erratum}, thus invalidating the conclusions of Ref.~\cite{cai2023quantum}.

Finally, we make a few remarks regarding the foundational issues surrounding the notion of quantum non-Markovianity and the dynamical map formalism of open quantum systems. Recent developments \cite{pollock2018non,milz2019completely} suggest that an \emph{operational} notion of (non-)Markovianity renders the two-time map approach described above insufficient in describing non-Markovian stochastic process in full generality. However, for our current purposes it suffices to note that the full map $\mathcal{E}(t_2,t_0)$ describing the overall timescale of the non-Markovian process is still a CPTP map and can be viewed as a black box. Furthermore, for the non-Markovian amplitude damping considered in this paper, there exists an exactly solvable time-dependent GKSL-type equation (\ref{eq:gksl-like}) thanks to conservation of photon numbers, which means that we have exact time-dependent Kraus operators at our disposal. In this light, given the full noise profile, we will see that the above discussed operational issues have no negative effects on channel-adapted QEC.

\section{Fidelity bounds on exact QEC for non-Markovian noise}\label{sec:bounds_on_exact}
We begin by studying the effects of non-Markovian noise within the framework of standard, general-purpose quantum error correction. Since the non-Markovian noise regime involves rapid decay and recurrence in the population and coherence levels in the density matrix, this in turn implies that the noise strength rapidly reaches its maximum strength within a short time, thus raising the question of the efficacy of QEC in this regime. We first write down an explicit expression for the worst-case fidelity, within the framework of \emph{perfect} QEC where some subset of errors are assumed to be perfectly correctable by the code.

\begin{lemma}
Consider a HPTP map $\cE[.]= \sum\limits_{k} {\rm sign}(k) A_k[.]A_k^{\dagger}$ and a codespace $\cC$ with associated projector $P$. Let $\{E_{i}\}$ denote the subset of errors perfectly correctable by the code $\cC$ and $\{F_{j}\}$ denote the subset of errors not correctable by the code $\cC$. The worst-case fidelity after QEC is then given by,

\begin{eqnarray}
    &&  F^2_{\rm min} = 1 - \nonumber \\
        && \underset{|\psi\rangle \in \cC}{\rm min} \sum\limits_{k, l} \frac{{\rm sign}(l)}{\alpha_{kk}} \left(\langle \psi|M^{\dagger}_{kl}M_{kl}|\psi\rangle - |\langle\psi|M_{kl}|\psi\rangle|^2\right), 
     \end{eqnarray}
with scalars $\{\alpha_{kk}\}$ and matrices $\{M_{kl}\}$ defined by,
\begin{eqnarray}
    PE_{k}^{\dagger}E_{l}P &=& \alpha_{kl}P, \nonumber \\
    PE_{k}^{\dagger}F_{l}P &=& PM_{kl}P. \label{eq:corr_err} 
\end{eqnarray}
\end{lemma}
\begin{proof}

We first rewrite the noise process $\cE$ in terms of the correctable and un-correctable errors as
\begin{align}
    \cE[.] &= \sum\limits_{k}{\rm sign}(k) E_{k}[.]E_{k}^{\dagger} +\sum\limits_{l}{\rm sign}(l) F_{l}[.]F_{l}^{\dagger}.
    \end{align}
For the correctable errors $\{E_k\}$, we can construct a CPTP recovery map $\cR$ with the Kraus operators $R_k = PU_{k}^{\dagger}$, where $\{U_k\}$ are unitaries appearing from the polar decomposition of the operators $E_kP= U_{k}\sqrt{PE_k^{\dagger}E_kP}$. Therefore, the recovery operators $R_k$ take the form,
    \begin{align}\label{eq:rec_op}
        R_k = \frac{1}{\sqrt{\alpha_{kk}}}PE_{k}^{\dagger}.
    \end{align}
The recovered state corresponding to an initial encoded state $|\psi\rangle\langle\psi|$ is thus given by,
    \begin{align}\label{eq:rec_state}
        \cR\circ\cE\,[|\psi\rangle\langle\psi|\,]&= \sum\limits_{k,l} \frac{{\rm sign}(l)}{\alpha_{kk}}P E_k^{\dagger}E_lP |\psi\rangle\langle\psi| E_l^{\dagger}E_{k}P \nonumber \\ &+  \sum\limits_{k,l} \frac{{\rm sign}(l)}{\alpha_{kk}}P E_k^{\dagger}F_lP |\psi\rangle\langle\psi| F_l^{\dagger}E_{k}P.
    \end{align}

From Eq.~\eqref{eq:rec_state}, we can derive the fidelity between the error-corrected state and the initial state as, 
\begin{align}\label{eq:fidelity}
        F^2\left(\cR\circ\cE [\,|\psi\rangle\langle\psi|\,,|\psi\rangle] \right)&= \sum\limits_{k,l} \frac{{\rm sign}(l) |\alpha_{kl}|^2}{\alpha_{kk}} \nonumber \\ 
        &+  \sum\limits_{k,l} \frac{{\rm sign}(l)}{\alpha_{kk}} |\langle\psi|M_{kl}|\psi \rangle|^2.
    \end{align}

\noindent This equality follows directly from Eqs.~\eqref{eq:corr_err}. The trace preserving (TP) condition for the noise process implies, 
     \begin{align}\label{eq:tp_noise}
         \sum\limits_{k} {\rm sign}(k)E_k^{\dagger}E_k+ \sum\limits_{k}{\rm sign}(k) F_k^{\dagger}F_k = I
     \end{align}
Using this and the TP condition $\sum\limits_{k} R_k^{\dagger}R_k = I$ for the recovery, we get,
     \begin{align}
          \sum\limits_{k,l} \left( {\rm sign}(k) E_k^{\dagger}R_l^{\dagger}R_lE_k + {\rm sign}(k) F_k^{\dagger} R_l^{\dagger}R_lF_k \right) = I.
     \end{align}
Using the form of the recovery operators in Eq.\eqref{eq:rec_op} and Eqs.~\eqref{eq:corr_err} the above equation takes the form,
     \begin{align}
      P &= P  \sum\limits_{k,l} \left( {\rm sign}(k) E_k^{\dagger}R_l^{\dagger}R_lE_k + {\rm sign}(k) F_k^{\dagger} R_l^{\dagger}R_lF_k \right)P, \nonumber \\
      &=  \sum\limits_{k,l} \frac{{\rm sign}(l) |\alpha_{kl}|^2 P}{\alpha_{kk}} + \sum\limits_{k,l} \frac{{\rm sign}(l) }{\alpha_{kk}} M_{kl}^{\dagger}M_{kl} \label{eq:alpha_kl_P}.
     \end{align}
     Therefore, using the Eq.~\eqref{eq:alpha_kl_P}, we can rewrite the fidelity expression in Eq.~\eqref{eq:fidelity} to get the desired expression for the worst-case fidelity,
     \begin{eqnarray}
         && F^2_{\rm min} = 1 - \nonumber \\
          && \underset{|\psi\rangle \in \cC}{\rm min}\sum\limits_{k, l} \frac{{\rm sign}(l)}{\alpha_{kk}} \left(\langle \psi|M^{\dagger}_{kl}M_{kl}|\psi\rangle - |\langle\psi|M_{kl}|\psi\rangle|^2\right).
          \label{eq:exactwcf-exactQEC}
     \end{eqnarray}
     \end{proof}
     
The fidelity loss, defined as $\eta = 1-F^2_{\rm min}$, therefore has the form,
\begin{align}\label{eq:fid_loss_perfect}
         \eta = \underset{|\psi\rangle \in \cC}{\rm min} \sum\limits_{k,l} \frac{{\rm sign}(l) }{\alpha_{kk}} \left(\langle \psi|M^{\dagger}_{kl}M_{kl}|\psi\rangle - |\langle\psi|M_{kl}|\psi\rangle|^2\right).
     \end{align}

From the Cauchy-Schwarz inequality we observe that the expression in the parenthesis in the Eq.\eqref{eq:fid_loss_perfect}, remains positive for all the time. The ${\rm sign}(l)$ term however, becomes negative for some error operators during those intermediate times for which there is information back flow. This clearly shows that the fidelity (or, correspondingly, the fidelity loss) exhibits an oscillatory behavior in the presence of non-Markovian noise, even within the framework of perfect QEC. 

We can see this explicitly, by noting that for certain time intervals, the infidelity  $\eta$ takes the following form,
     \begin{align}
         \eta &= \underset{|\psi\rangle \in \cC}{\rm min} \sum\limits_{k,l_{1}} \frac{1}{\alpha_{kk}} \left(\langle \psi|M^{\dagger}_{kl_{1}}M_{kl_{1}}|\psi\rangle - |\langle\psi|M_{kl_{1}}|\psi\rangle|^2\right) \nonumber\\ 
         &\qquad-\underset{|\psi\rangle \in \cC}{\rm min} \sum\limits_{k,l_{2}} \frac{1}{\alpha_{kk}} \left(\langle \psi|M^{\dagger}_{kl_{2}}M_{kl_{2}}|\psi\rangle - |\langle\psi|M_{kl_{2}}|\psi\rangle|^2\right),
     \end{align}
where ${\rm sign}(l_{1}) > 0$ and ${\rm sign}(l_{2}) < 0$. In the above expression the second term corresponds to a drop in the fidelity loss and hence a revival of fidelity. The first term, on the other hand, corresponds to damping of the peaks in the revivals. We demonstrate this with an example in  Sec.~\ref{sec:examples}, while performing the stabilizer-measurement based recovery for non-Markovian amplitude-damping noise. 

Finally, one may easily verify that as the noise strength approaches its minimum value, the fidelity loss reaches its minimum. In fact, it was previously shown that any stabilizer code would meet a similar end for non-Markovian decoherence if a stabilizer recovery is performed in a continuous time QEC setting~\cite{oreshkov2007continuous}. In what follows, we show how our work points to the possibility of error correcting non-Markovian noise without having the fidelity drop to zero at finite times.

\section{Noise-adapted QEC for non-Markovian errors}\label{sec:noise-adapted}
We now consider two kinds of noise-adapted recovery maps tailored to non-Markovian noise and characterize their performance via the worst-case fidelity. Note that since HPTP errors are not accessible experimentally during intermediate times when the map in non-CP, a noise-adapted recovery for HPTP maps is not practically implementable. However, in what follows, we will show how the Petz recovery map adapts to the \textit{full} map $\mathcal{E}(t_2,t_0)$ and corrects for all intermediate HPTP errors in such a way that it takes advantage of the information backflow occurring naturally due to non-Markovian decoherence over a period of time. The recoverability of Petz channel in such a case can be attributed to the fact that Petz-based recovery does not rely on syndrome extraction via a feedback loop.

One might be led to assume that a recovery operation such as the Petz channel, when adapted to a non-Markovian noise, must also be non-Markovian in order to achieve a good fidelity. However, as we show later, even a Markovian Petz channel suffices for recovery with pretty good fidelity, although in this case the Petz channel does not exactly adapt to the noise in the sense that it adapts to the Kraus operator structure but not to the time-dependent noise parameter. The price to pay, however, is that the total $\mathcal{R}\circ\mathcal{E}$ channel need not be unital. The worst-case fidelity therefore suffers a drop, an aspect which is further discussed in Sec. \ref{sec:nonunital-consequence}.

Recall from Eq.~\eqref{eqn:code-specific-petz} that for a noise channel with Kraus operators $\{E_{i}\}$ and a codespace with projector $P$, the Petz channel is the map with Kraus operators $R_i = PE_i^\dagger(t) \cE[P]^{-1/2}$. In the case of non-Markovian noise, we may consider two variants of the Petz recovery map. Naively, we could let the Petz recovery operators exactly adapt to the noise operators at all times. This would result in a \emph{non-Markovian Petz channel}, denoted as $\cR^{NM}$, one that violates the CP-divisibility condition in Eq.~\eqref{eq:divisibility-condition} ($\cR(t_2,t_0)=\cR(t_2,t_1)\cR(t_1,t_0)$) and has an operator sum-difference representation of the form,
\begin{equation}
    \cR^{NM}[\rho] = \sum_{j} {\rm sign}(j)R_{j} [\rho] R_{j}^{\dagger}. \label{eq:nmPetz}
\end{equation}
On the other hand, we could construct a \emph{Markovian Petz channel} (denoted as $\cR^{M}$) where the operators $R_{i}$ are adapted only to the structure of the noise operators $E_{i}$, but with a \emph{fixed} noise strength so as to ensure that the map stays in the Markovian regime for all times. In this case, the Petz recovery map is CP-divisible in the sense of Eq.~\eqref{eq:divisibility-condition}.

We now obtain bounds on the worst-case fidelity using both the Markovian and non-Markovian Petz recovery maps via a simple generalization of the AQEC conditions in Ref.~\cite{ng2010simple} to HPTP noise channels.
\begin{theorem}
Consider a HPTP map $\cE(t)\sim\{{\rm sign}(i),E_i(t)\}$, and a $d$-dimensional code $\cC$ with projector $P$. Let $\Delta_{ij}(t)\in\cB(\cC)$ be traceless operators such that
\begin{equation}\label{eq:AECcond1}
PE_i^\dagger (t) \cE[P]^{-1/2}E_j(t) P=\beta_{ij}(t)P+\Delta_{ij}(t),
\end{equation}
where $\beta_{ij}(t)\in\mathbb{C}$. Then, the fidelity loss achieved using the non-Markovian Petz recovery is given by,
\begin{align}
\eta(t)& = \sum\limits_{i,j}\, {\rm sign}(i) {\rm sign}(j) (\langle \psi|\Delta_{ij}^{\dagger}(t)\Delta_{ij}(t)|\psi\rangle \nonumber\\
& \qquad-|\langle \psi|\Delta_{ij}(t)|\psi\rangle|^2).
\label{eq:etaP2}
\end{align}
For the Petz recovery adapted to the Markovian regime of the noise, the fidelity loss takes the form,
\begin{align}
    \eta(t) &= \sum\limits_{i,j}\, {\rm sign}(i)(\langle \psi|\Delta_{ij}^{\dagger}(t)\Delta_{ij}(t)|\psi\rangle -|\langle \psi|\Delta_{ij}(t)|\psi\rangle|^2). \label{eq:markov-loss-Petz}
\end{align}
\end{theorem}
\begin{proof}
The proof is straightforward. From Eq.~(\ref{eq:hptp-kraus}) we have $\sum\limits_{i}\, {\rm sign}(i)\,E_i^{\dagger}(t)E_i(t)= I$. Note that, for a Petz recovery with Kraus operators $R_i(t)$, which are \emph{exactly} adapted to the non-Markovian noise process, the completeness condition is given by $\sum\limits_{i}\, {\rm sign}(i)\,R_i^{\dagger}(t)R_i(t)= I$. Therefore, we have, 
    \begin{align}
          \sum\limits_{i,j}\, {\rm sign}(i) {\rm sign}(j) \,E_i^{\dagger}(t)R_j^{\dagger}(t)R_j(t)E_i(t)= I.
    \end{align}    
Now, using the AQEC condition in Eq.\eqref{eq:AECcond1}, we get, 
\begin{align}\label{eq:tp_petz}
    \sum\limits_{i,j}\, {\rm sign}(i){\rm sign}(j) \,(|\beta_{ij}(t)|^2 +\langle\psi|\Delta_{ij}^{\dagger}(t)\Delta_{ij}(t)|\psi\rangle \nonumber \\+\langle\psi|\Delta_{ij}^{\dagger}(t)|\psi\rangle\beta_{ij}(t) +\beta^*_{ij}(t)\langle\psi|\Delta_{ij}(t)|\psi\rangle)&=1.
\end{align}
Therefore, the fidelity between the recovered state $(\cR^{NM}\circ \cE)(|\psi\rangle\langle\psi|)$, and the input state $|\psi\rangle$ works out to be,
\begin{align}
 F^2 &=  \sum\limits_{ij} {\rm sign}(i) {\rm sign}(j) \langle\psi|PE_i^{\dagger}\cE[P]^{-1/2}E_jP|\psi\rangle|^2.
\end{align}
Finally, using the identity in Eq.\eqref{eq:tp_petz} and Eq.\eqref{eq:AECcond1}, we arrive at the expression for the infidelity $\eta$ between the recovered state and the input state $|\psi\rangle \in \cC$, given in Eq.~\eqref{eq:etaP2}.

The proof for the fidelity bounds for Markovian Petz recovery follows in similar lines of the previous case, except with ${\rm sign} (j) = +1$ for all times $t$, because the Petz channel is now adapted to the noise in the Markovian regime and is CP-divisible: $\cR^{M}(t_2,t_0) = \cR^{M}(t_2,t_1)\cR^{M}(t_1,t_0)$, for all $t_2 \ge t_1 \ge t_0$.
\end{proof}
Eq.~\eqref{eq:etaP2} elucidates how the fidelity-loss arises from the presence of the $\Delta_{ij}$ operators. If $\Delta_{ij}=0~\forall i,j$, the above conditions reduce to exact (or, standard) QEC, and the Petz recovery map reduces to the standard recovery operation. 

Interestingly, we note that the near optimality of the Petz recovery is independent of whether the errors are Markovian or not, in the sense that it is near-optimal for \textit{any} CPTP noise and the \textit{full} non-Markovian map $\mathcal{E}(t_2,t_0)$ is CPTP, although it is CP-\textit{in}divisible. The proof of near-optimality follows from the previous paper \cite{ng2010simple}, which we do not repeat here. Rather, we demonstrate the near-optimality of the Petz map for the specific case of non-Markovian amplitude-damping noise.

\section{The case of  non-Markovian amplitude damping noise \label{sec:examples}}

We consider a prototypical model of non-Markovian noise, namely, the non-Markovian amplitude damping {(AD)} channel which arises from a damped Jaynes-Cummings (JC) interaction of a two-level system and a bosonic reservoir at zero temperature. This model has a well defined Markovian and non-Markovian regime, as explained in  Appendix \ref{sec:NMAD}. Here, we directly write down the solution to the master equation in the form of the Kraus operators. The action of the non-Markovian AD noise on a density operator $\rho(t)$ is given by $\rho(t) = \mathcal{E}(t)[\rho]  = \sum_j E_j(t) \rho E^\dagger _j (t)$, where the \emph{time-dependent} operators $E_{1}(t)$ and $E_{2}(t)$ are given by,
\begin{align}
E_1(t) = \left(
\begin{array}{cc}
1 & 0 \\
0 & \sqrt{1-\gamma(t)} \\
\end{array}
\right);  E_2(t) = \left(
\begin{array}{cc}
0 & \sqrt{\gamma(t)} \\
0 & 0 \\
\end{array}
\right).
\label{eq:adkraus1}
\end{align}
For the damped JC model considered here, the damping parameter $\gamma(t)$ takes the form \cite{breuer2016colloquium} $$\gamma(t)=1-|G(t)|^2,$$ with
 \begin{align}
G(t) =e^{-\frac{b t}{2}} \left( \frac{b}{d} \sinh \left[\frac{d t}{2}\right] + \cosh \left[\frac{d t}{2}\right] \right).
\label{eq:G}
\end{align}
Here, $d = \sqrt{b^2 - 2\Gamma_0 b}$, where $\Gamma_0$ quantifies the strength of the system-environment coupling and $b$ is the spectral bandwidth. The system undergoes non-Markovian evolution when $b << 2\Gamma_0$. When $b >> 2\Gamma_0$, the dynamics is time-homogeneous Markovian -- it has a Lindblad form with a constant decay rate ($\Gamma(t) = \Gamma_0$) and the corresponding channel belongs to a family of one-parameter semigroups.

For this noise model, we now study the performance of standard QEC in comparison {with noise}-adapted AQEC strategies. {Specifically,} we consider {one example of a} stabilizer and {a} noise-adapted code in combination with three different types of recovery operations, namely the standard syndrome-based recovery as well as the noise-adapted Leung~\cite{leung} and Petz recoveries. {For each of these cases, we numerically obtain the worst-case fidelity as a function of time and summarize our results in the plots shown in Figs.~\ref{fig:plot1} and \ref{fig:nm_leung}.}

First, we consider the standard {approach to} QEC based on the $[[5,1,3]]$ code, which has the following stabilizer generators.
\begin{align}
    S = \langle XZZXI,IXZZX,XIXZZ,ZXIXZ \rangle . 
\end{align}
{We see from Figure~\ref{fig:plot1}, that when this code is used in conjunction with the syndrome-based recovery (denoted as $\cR^{S}$) to correct for non-Markovian amplitude damping noise, the} worst-case fidelity {shows oscillatory behavior with time, following the expressions in Eq.~(\ref{eq:fid_loss_perfect}). More importantly, it reaches a minimum value that is} below $0.5$, {in finite time}. This can be explained by noting that in the large damping limit the uncorrectable set of errors have {a} greater contribution {to} the fidelity, which {can be exactly evaluated to be,} 
\begin{equation}
    F^2_{\rm min}[\cR^{S}\circ\cE] = 1-1.875 (\gamma^2 - \gamma^3) -0.625 \gamma^4.
\end{equation} 
Clearly, when $\gamma = 1$ one can observe that $F^2_{\rm min} \approx 0.375$. This is indeed consistent with the conclusions drawn in Ref.~\cite{oreshkov2007continuous} {regarding the performance of the} three-qubit code {with} stabilizer recovery, {in the presence of non-Markovian bit-flip noise}.

\begin{figure}[t!]
    \centering
    \includegraphics[width=1\columnwidth]{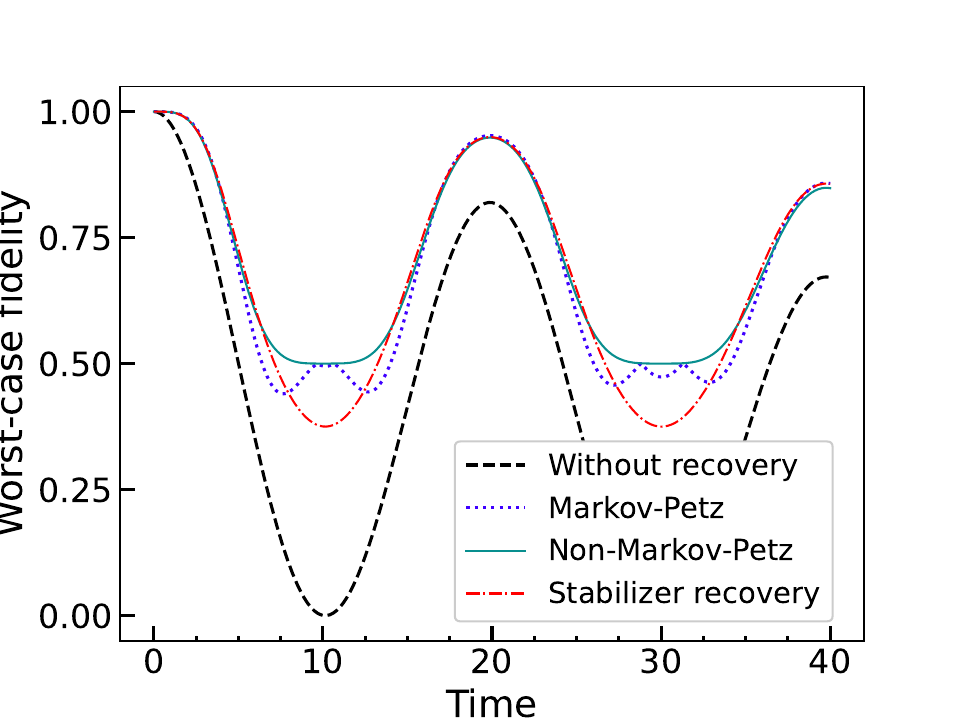}
    \caption{Performance of different QEC schemes for non-Markovian AD noise with noise parameters $b=0.01$ and $\Gamma_0=5$ (see Appendix \ref{sec:NMAD}). We compare {the bare qubit fidelity with that achieved by the $[[5,1,3]]$ code with} stabilizer recovery and {the $4$-qubit code} with Petz recovery (Markovian and non-Markovian).}
    \label{fig:plot1}
\end{figure}

Next, we consider the four-qubit code tailored to correct {for} amplitude damping errors, {with} codewords~\cite{leung}, 
\begin{align}
\ket{0_L}&= \frac{1}{\sqrt{2}}(\ket{0000} + \ket{1111}) \nonumber \\
\ket{1_L}&= \frac{1}{\sqrt{2}}(\ket{0011} + \ket{1100}). \label{eq:4-qubit}
\end{align}

\noindent We study the performance of the Petz recovery map defined in Eq.~\eqref{eqn:code-specific-petz}, adapted to this $4$-qubit code in the presence of non-Markovian amplitude damping noise. We consider both the Markovian and non-Markovian versions of the Petz map, as described in Sec.~\ref{sec:noise-adapted}. We find that the Markovian and the non-Markovian Petz recoveries outperform the $5$-qubit stabilizer code and syndrome-based recovery, as evidenced by the plots in Fig.~\ref{fig:plot1}. 

In fact, we see from Fig.~\ref{fig:plot1} that the non-Markovian Petz channel safeguards the codespace {with a fidelity greater than $0.5$, even at the maximum noise limit corresponding to time $t=10$ units. We can obtain an analytical expression that provides the best fit for the worst-case fidelity, for the $4$-qubit code under Markovian and non-Markovian Petz recovery, respectively, as,
\begin{eqnarray}
   F^2_{\rm min}[\cR^{M}\circ\cE] =  1 &-& 1.658 \gamma^2 + 1.069 \gamma^3 - 1.517 \gamma^4 \nonumber \\
   &+& 2.563 \gamma^5 -0.955 \gamma^6, \nonumber \\ 
 F^2_{\rm min}[\cR^{NM}\circ\cE] = 1 &-& 1.715 \gamma^2 + 0.362 \gamma^3 + 2.35 \gamma^4 \nonumber \\
 &-&   1.93 \gamma^5 +0.428 \gamma^6.
\end{eqnarray} 

However, when the Petz channel is adapted to the Markovian regime of the noise, we observe a drop in the fidelity below $0.5$. We show in Sec.~(\ref{sec:nonunital-consequence}) below that this drop occurs due the fact that the composite QEC channel $\cR^{M} \circ \cE$ becomes non-unital in this case. This happens because the Markovian Petz map is \emph{not exactly} adapted to the noise process -- its Kraus operators are adapted to the structure of the Kraus operators of the noise channel but not to the \emph{exact} noise strength.

 \begin{figure}[t!]
 \centering
\includegraphics[width=1\columnwidth]{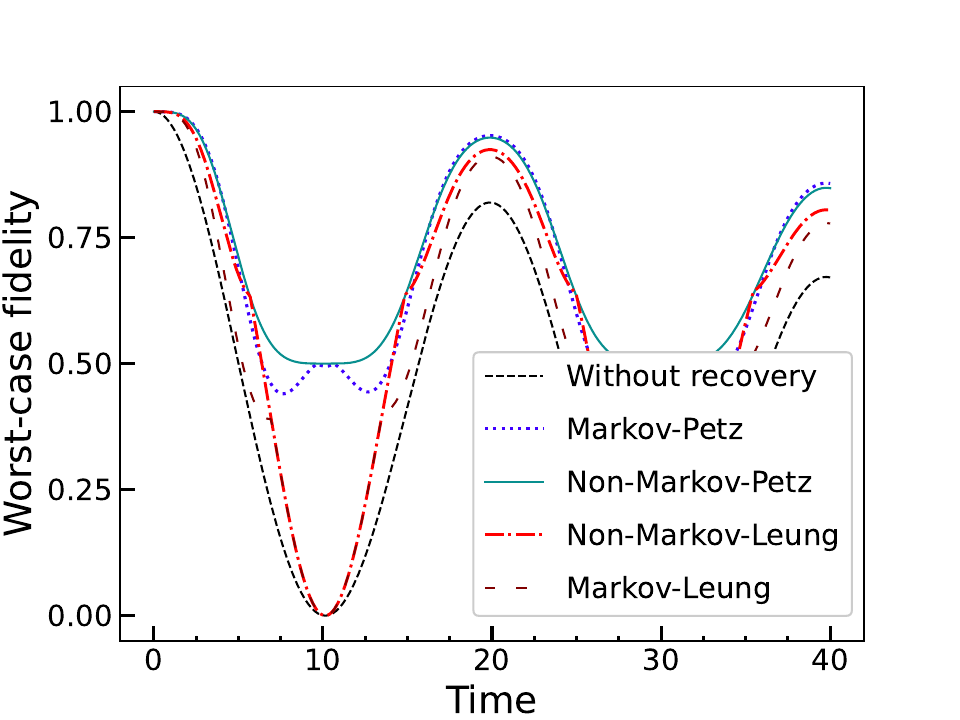} 
\caption{Performance of different recovery schemes for the $[4,1]$ code subject to non-Markovian AD noise with the noise parameters $b=0.01$ and $\Gamma_0=5$ (see Appendix \ref{sec:NMAD}). For the Markovian Petz and Markovian Leung recovery, the recoveries are adapted to AD noise in the Markovian regime with $b=0.1$ and $\Gamma_0 =0.005$.}
\label{fig:nm_leung}
\end{figure} 

In Fig.~\ref{fig:nm_leung}, for the $4$-qubit code adapted to AD noise, we compare the performance of two different noise-adapted recoveries, namely the Petz map and the recovery map constructed by Leung \emph{et al.}~\cite{leung}. The Leung recovery map involves nontrivial unitaries that arise via polar decomposition of the Kraus operators of the noise map. We merely define the Leung recovery here and refer to Appendix~\ref{sec:leung_reco_app} for a detailed discussion. For a code $\cC$ with projector $P$ that satisfies the set of AQEC conditions in Eqs.~\eqref{eq:ortho_nm_leung} for a set of noise operators $\{E_{i}\}$, the Leung recovery has Kraus operators $R_{i} = PU_{i}^{\dagger}$, where $E_{i}P = U_{i}P\sqrt{PE_{i}^{\dagger}E_{i}P}$. It was shown~\cite{leung} that this recovery map achieves the desired order of fidelity for any code and noise process that satisfies Eqs.~\eqref{eq:ortho_nm_leung}.

\begin{widetext}
\begin{figure*}[t]
  \includegraphics[width=0.62\textwidth]{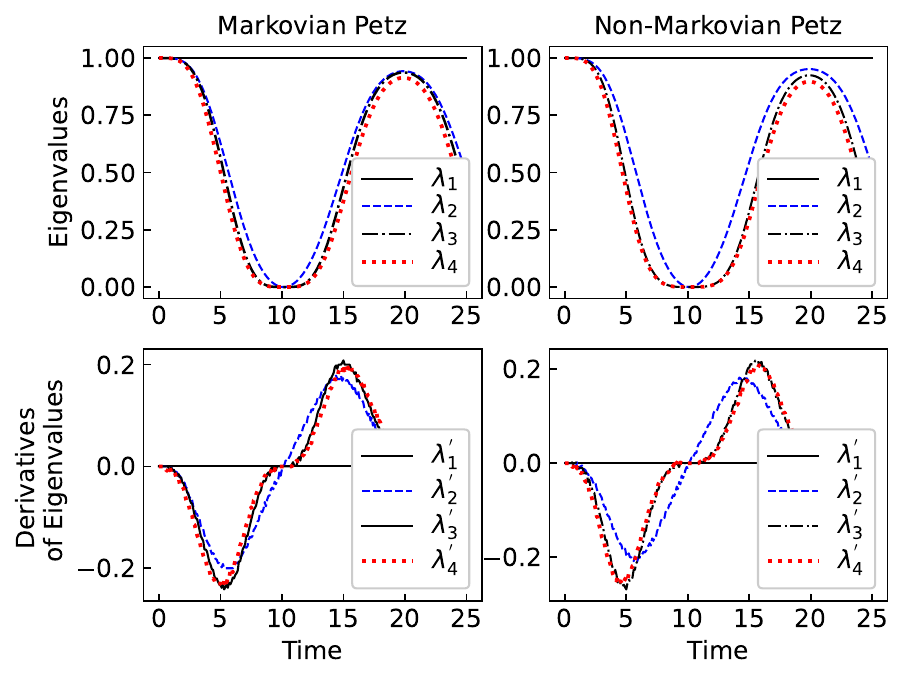}
      \caption{Eigenvalues of the matrix $M$ in Eq.\eqref{eq:RE-matrixform} (\emph{top graphs}) and their time-derivatives ($\lambda'$) (\emph{bottom graphs}) for Markovian Petz channel (adapted to AD with parameter $b=0.1$ and $\Gamma_0=0.005$) and non-Markovian Petz channel (adapted to AD with parameter $b=0.01$ and $\Gamma_0=5$) specific to the $4$-qubit code. The noise parameters are explained in Appendix \ref{sec:NMAD}.}
      \label{fig:eig1}
\end{figure*}    
\end{widetext}

The $4$-qubit code in Eq.~\eqref{eq:4-qubit} satisfies the AQEC conditions in Eq.~\eqref{eq:ortho_nm_leung} for the amplitude damping channel. We can therefore study the performance of the Leung recovery adapted to both Markovian and non-Markovian regimes of the AD channel. However, as seen in Fig.~\ref{fig:nm_leung}, both Markovian as well as non-Markovian Leung recoveries are ineffective in the large damping limit, since the fidelity becomes zero. The poor  performance of the Leung recovery map can be understood  via the fidelity expressions derived in Appendix~\ref{sec:leung_reco_app}.

\subsection{Consequences of the non-unitality of the QEC channel \label{sec:nonunital-consequence}}
Here, we explain the early recurrence of worst-case fidelity, before the point in time at which information backflow begins, observed in Figs.~\ref{fig:plot1} and~\ref{fig:nm_leung}, for the Markovian Petz recovery map. This additional oscillation in the fidelity is due to the combined effect of the structure of the code and the non-unitality of the QEC channel and is not due to non-Markovianity. This is  demonstrated in Appendix \ref{sec:inexact-demo}, by showing such a revival even in the case of Markovian noise, when the Petz recovery is not exactly adapted to the noise strength. On the other hand, when the Petz recovery channel is non-Markovian, which corresponds to exactly  adapting to the noise, the composite $\mathcal{R}^{NM}\circ \mathcal{E}$ channel is unital. These conclusions become apparent when viewed via the matrix representation of a CPTP map as follows.

We consider the composite CPTP map $\mathcal{R}\circ \mathcal{E} \equiv \Phi$. The action of $\Phi$ on a density operator $\rho$, can be represented by a matrix $M$ with the matrix elements 
\begin{align}\label{eq:M_bloch}
    M_{ij}&={\rm Tr}[O_i \Phi[O_j]],
\end{align}
where $\{O_i\}_{i=0}^{d^2-1}$ forms an Hilbert-Schmidt basis. 

The matrix $M$ is real because the the operators $O_i$ are hermitian and the map $\Phi$ is CP. We can write any state $\ket{\psi} \in \cC$ in the Hilbert-Schmidt basis as 
\begin{align}\label{eq:rho_bloch}
    \rho = \ket{\psi}\bra{\psi} = \frac{1}{d}(I +\Vec{O}.\Vec{r}),
\end{align}
where $\Vec{O}= (O_1,O_2,\hdots,O_{d^2-1})^T$ and $\Vec{r}$ is a real $(d^2-1)$ element Bloch-vectors, and $T$ is the transpose operation. As prescribed in \cite{hkn_pm2010}, the  Eq.\eqref{eq:rho_bloch} along with the Eq.\eqref{eq:M_bloch}, gives the following expression for the fidelity between the states $\Phi[\ket{\psi}\bra{\psi}]$ and $\ket{\psi}$ as 
\begin{align}\label{eq:fid_M}
    F^2 &= \frac{1}{d}(r^{\rm T}. M. r).
\end{align}
Therefore, for a code $\cC$ of dimension $d$, the matrix $M$ becomes a $d\times d$ matrix with following form 
\begin{align}
 M = \left(
\begin{array}{c|c}
1&0\\
\hline
\Vec{\tau}& \quad T_{{d^2-1}\times {d^2-1}}\\
\end{array}
\right).
\label{eq:RE-matrixform}
\end{align}
 Note that if the channel is unital onto the codespace $\cC$, i.e., $\Phi[P]= P$, then $\Vec{\tau}= \Vec{0}$. Therefore in that case the fidelity will depend upon the eigenvalues of the matrix $T_{{d^2-1}\times {d^2-1}}$ \cite{hkn_pm2010}. However the fidelity for a non-unital channel is 
\begin{align}\label{eq:eq:fid_nonuni}
    F^2 &= \frac{1}{d} (r^{\rm T}. T . r + \Vec{\tau}.\Vec{r}).
\end{align}
Now for a qubit code ($d=2$), the operators $O_i$s are nothing but the Pauli operators onto the codespace \cite{hkn_pm2010} and $O_0= P$ -- the projector onto $\cC$. We estimate the worst-case fidelity $F^2_{\rm min}$ from the Eq.\eqref{eq:eq:fid_nonuni} for the $[4,1]$ code considering $\Phi$ under the constraint $r_x^2+r_y^2+r_z^2=1$. 

Although non-unitality of the composite QEC channel causes the fidelity to drop below the worst-case limit, it is pertinent to ask if combining Petz recovery channel with codes larger than the [[4,1]] code would provide certain improvement. In order to verify the performance of this scheme, we consider [[5,1,3]] and [[6,1,3]] codes with Petz recovery channel and interestingly find that the former performs better than the latter. We refer the reader to Appendix \ref{sec:5-6-code-remedy} for the details. 

\subsection{Non-Markovianity of the composite  QEC channel}\label{sec:nm_pdiv}
Finally, we discuss how we may characterize the non-Markovianity of the entire QEC process via a correspondence between the non-monotonous behavior of fidelity and its relation with indivisibility of $\mathcal{R}\circ \mathcal{E}$ channel. Following \cite{chruscinski2017detecting}, a Hermitian dynamical map is said to be P-divisible iff
\begin{align}
    \frac{d}{dt}\lambda_k(t) \le 0,
    \label{eq:P-div}
\end{align}
where $\lambda_k(t)$ is the $k$-th eigenvalue of the matrix $M$ given in Eq.~(\ref{eq:RE-matrixform}). In Figs.~\ref{fig:eig1} and ~\ref{fig:stab_eigs}, we plot the eigenvalues of the composite channel for the three different QEC schemes studied in Fig.~\ref{fig:plot1}, namely,  the $4$-qubit code with Markovian and non-Markovian  Petz recoveries and the $[[5,1,3]]$ code with stabilizer recovery. Interestingly, for the $[[5,1,3]]$ code, only one eigenvalue is non-constant and contributes to the violation of the inequality Eq.~\eqref{eq:P-div}, whereas for the $4$-qubit code, all the eigenvalues are non-constant and contribute to the violation of Eq.~\eqref{eq:P-div}, as seen in  Fig.~\ref{fig:eig1}.  

\section{Conclusion}\label{sec:concl}
We study the problem of quantum error correction for non-Markovian noise, specifically focusing on the class of Hermiticity-preserving trace-preserving (HPTP) maps which are not completely positive (CP) for certain intermediate times. We first provide fidelity bounds for both standard QEC as well as noise-adapted QEC for this class of noise maps. 

We then consider a specific example of a physically motivated noise model, namely a non-Markovian amplitude damping channel, and compare the performance of exact and noise-adapted QEC schemes for this noise process. We see that the non-Markovian Petz channel adapts to the entire noise map and effectively corrects all intermediate HPTP errors in the low noise limit and preserves the code space in the highest noise limit. In other words, the Petz channel adapted to a fixed timescale of non-Markovian noise leverages information backflow. It uniquely safeguards the worst-case fidelity against the rapid fluctuations of information flow between the system and environment, a capability not shared by other recovery operations.  
Our findings further strengthen the known results on the near-optimality of the Petz channel for Markovian noise~\cite{ng2010simple}. 
\begin{figure}
    \centering
    \includegraphics[width=1\columnwidth]{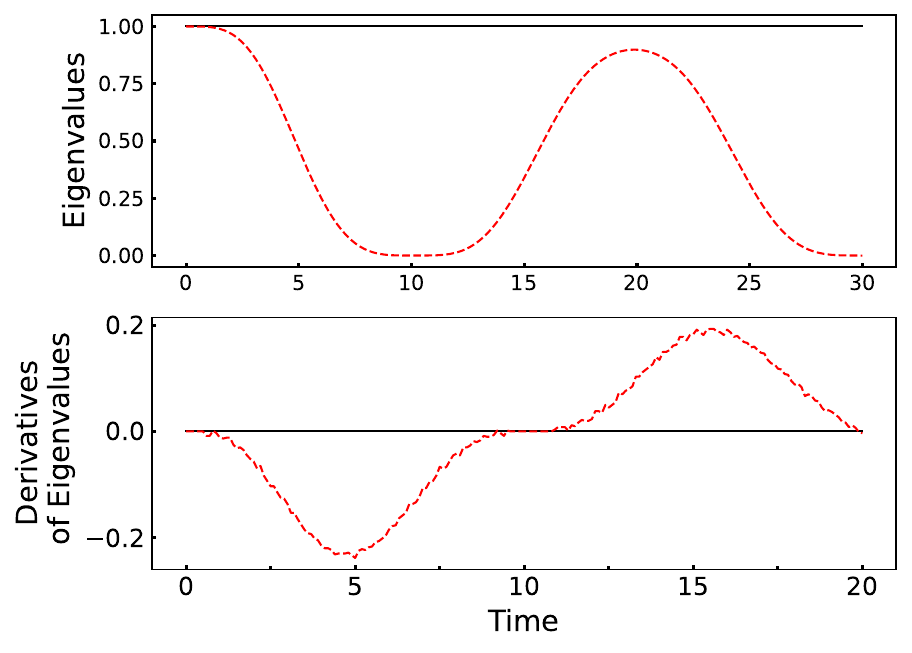}
    \caption{Eigenvalues of the $M$ matrix (top graph) and their derivatives (bottom graph), where the recovery is the stabilizer recovery for the $[[5,1,3]]$ code and the noise is the non-Markovian AD noise (see Appendix \ref{sec:NMAD}), with $b=0.01$ and $\Gamma_0=5$. This plot shows only the non-zero eigenvalues of the matrix $M$.}
    \label{fig:stab_eigs}
\end{figure}

However, implementing such a non-Markovian Petz channel could be nontrivial and somewhat resource intensive. In fact, writing down a quantum circuit for CP-indivisible channel remains an open problem in the literature to our knowledge. Here, we overcome this hurdle by showing that a Petz channel can be inexactly adapted, meaning adapted to only the Kraus operator structure but not to the time-dependent damping parameter, so that a CP-divisible Petz channel can correct for CP-indivisible amplitude damping channel. However, the price we pay is that the total channel seizes to be unital causing a small reduction in fidelity, for a brief period of time, below the worst-case limit of $0.5$.

Going forward, a couple of open questions are in order. In this work, we considered the cases of the $4$-qubit code and the $[[5,1,3]]$ code with the Petz channel as recovery.  
It would be an interesting open question to explore the relationship between the eigenvalue spectra of the QEC superchannel and the optimality of the code to correct non-Markovian errors in general. It would be also interesting to further explore the relationship between the distance of a code and non-Markovianity of QEC superchannel \cite{mottonen2023quantum}. Finally, it would be an interesting challenge to study the efficacy of the Petz recovery map beyond the class of CP-indivisible maps and work towards a general theory of noise-adapted QEC for arbitrary forms of non-Markovian noise.

\section*{Acknowledgments}
SU thanks IIT Madras for the support through the Institute Postdoctoral Fellowship. This work was partially funded by a grant from the Mphasis F1 Foundation to the Center for Quantum Information, Communication, and Computing (CQuICC), IIT Madras. 


\section*{Appendix}
\appendix
\section{Non-Markovian amplitude damping channel \label{sec:NMAD}}
Let us consider damped Jaynes-Cummings (JC) model \cite{garraway1997nonperturbative,breuer2016colloquium} of a two level system interacting with a dissipative bosonic reservoir at zero temperature. The total system-reservoir Hamiltonian, in the RWA, is
\begin{align}
H_{\rm  tot} &= \frac{\omega_0 \sigma_z}{2} + \sum_j \omega_j a_j^\dagger a_j + \sum_j (g_j \sigma_{+} a_j + g^\ast _j \sigma_{-} a^\dagger _j),
\end{align}
where $a_j^\dagger$ and $a_j$ are creation and annihilation operators, and $g_j$ are coupling parameters and $\sigma_{+}=\vert 1 \rangle \langle 0 \vert$ and $\sigma_{-}=\vert 0 \rangle \langle 1 \vert$. And the total S-E evolution is given by the von-Neuman equation 
\begin{align}
   \frac{d}{dt} \rho_{\rm SE} = - \frac{i}{\hbar} [H_{\rm tot}, \rho_{\rm SE}]. \label{eq:totalevolution}
\end{align}
The above model has a Lorentzian environmental spectral density \cite{garraway1997nonperturbative}
\begin{align}
    J(\omega) = \frac{\Gamma_0 b^2}{2 \pi [(\omega_0 - \Delta - \omega)^2 + b^2]},  
    \label{eq:spectral}
\end{align}
where $ \Delta = \omega-\omega_0 $ is the detuning parameter which governs the shift in frequency $\omega$ of the qubit from the central frequency $ \omega_0$ of the bath. $ \omega_0$ represents the energy gap between ground state $\ket{0}$ and exited state $\ket{1}$. $\Gamma_0$ quantifies the strength of the system-environment coupling and $b$ the spectral bandwidth. Note that the spectral density in Eq.(\ref{eq:spectral}) leads to \textit{colored} memory kernel, giving rise to non-Markovian dynamics. This particular model indeed has well-defined Markovian and non-Markovian regimes, thus allowing manipulation of control parameters in the Kraus representation as follows.

Assuming that the environment starts out in the grounds state, tracing out the environment degrees of freedom from Eq.~(\ref{eq:totalevolution}) in the interaction picture, gives rise to the GKSL-like equation (\ref{eq:gksl-like}) for the model considered with single jump operator $L_j = L = \sigma_{-}$ \cite{garraway1997nonperturbative,breuer2002theory}:
\begin{equation}
    \frac{d \rho(t)}{dt} = \Gamma(t) \bigg(\sigma_{-} \rho(t) \sigma_{+} - \frac{1}{2}\{\sigma_{+} \sigma_{-}, \rho(t)\} \bigg) .
\label{eq:gksl-AD}
\end{equation} 
The solution to the above master equation then is given by an amplitude damping (AD) channel with the time-\textit{dependent} Kraus operators 
\begin{align}
E_1(t) = \left(
\begin{array}{cc}
1 & 0 \\
0 & \sqrt{1-\gamma(t)} \\
\end{array}
\right) \quad;\quad  E_2(t) = \left(
\begin{array}{cc}
0 & \sqrt{\gamma(t)} \\
0 & 0 \\
\end{array}
\right),
\label{eq:adkraus}
\end{align}
so that the action of the channel on a density operator $\rho(t)$ is given by $\rho(t) = \mathcal{E}(t)[\rho]  = \sum_j E_j(t) \rho E^\dagger _j (t)$.

The density matrix of the open system after the action of the channel is given by \cite{bellomo2007non}
\begin{equation}
			\rho_S(t)=\left(
			\begin{array}{cc}
				\rho_{11}(0)G(t)  & \rho_{10}(0)\sqrt{G(t)}\\\\
				\rho_{01}(0)\sqrt{G(t)}  & \rho_{00}(0)+ \rho_{11}(0)(1-G(t)) \\
			\end{array}\right),
		\end{equation}
		where the function $G(t)$ obeys the differential equation
		\begin{equation}\label{equforp}
			\frac{d}{dt}G(t) =-\int_0^t \mathrm{d} \tau f(t-\tau)G({\tau})\,,
		\end{equation}
		and the correlation function $f(t-\tau)$ is related to the spectral
		density $J(\omega)$ of the reservoir by
		\begin{equation}\label{corrfunc}
			f(t-\tau)=\int \mathrm{d} \omega J(\omega)\exp
			[i(\omega_0-\omega)(t-\tau)]\,.
		\end{equation}
 The time-dependent decay rate given in Eq.~(\ref{eq:gksl-AD}) is given by $\Gamma(t) = -\frac{2}{G(t)} \frac{d G(t)}{dt}$, where negative decay rates indicates non-Markovian dynamics. 

\section{Demonstration of inexact adaptation of Petz channel \label{sec:inexact-demo}}
As noted before, by inexact adaptation we mean that the recovery channel adapts only to the Kraus operator structure of the noise and not to the noise strength. In other words, while constructing the recovery, the noise strength can be arbitrarily chosen. Since our goal is to demonstrate the consequence of inexact adaptation, we do so only in the Markovian regime of the noise. For that, we consider the Petz recovery again for the five-qubit code and for the four-qubit Leung code tailored for correcting amplitude-damping noise. 

To study the performance of the Petz recovery, specific to the $4$-qubit code, we make it adapted to an amplitude-damping process with fixed noise strength, say $\gamma=0.1$, and vary the noise strength from $\gamma=0$ to $\gamma=1$ and plot the worst-case fidelity against it. We see that after $\gamma=0.86$ the worst-case fidelity increases until $\gamma=0.97$, and after that, fidelity falls to the minimum value, which is around $0.46$.

For the $[[5,1,3]]$ code, also known as the Laflamme code \cite{laflamme1996perfectcode, Bennett_david_1996}, we do not see such an increase in the fidelity value. However, if we adapt the Petz recovery to the amplitude damping with a different noise strength $\gamma=0.1$, we see different behavior for this code. As shown in Fig.~\ref{fig}, for such a noise-adapted recovery, the fidelity values approach a value around $\sim 0.48$, which is slightly less than $1/2$. Here, $1/2$ corresponds to the worst-case fidelity between the input state $\ket{\psi} \in \mathcal{C}$ and the maximally mixed state, which is the fixed point for the channel $\mathcal{R} \circ \mathcal{E}$ when $\mathcal{R}$ is exactly adapted to the noise process.

\begin{figure}
\centering
\includegraphics[width=1\columnwidth]{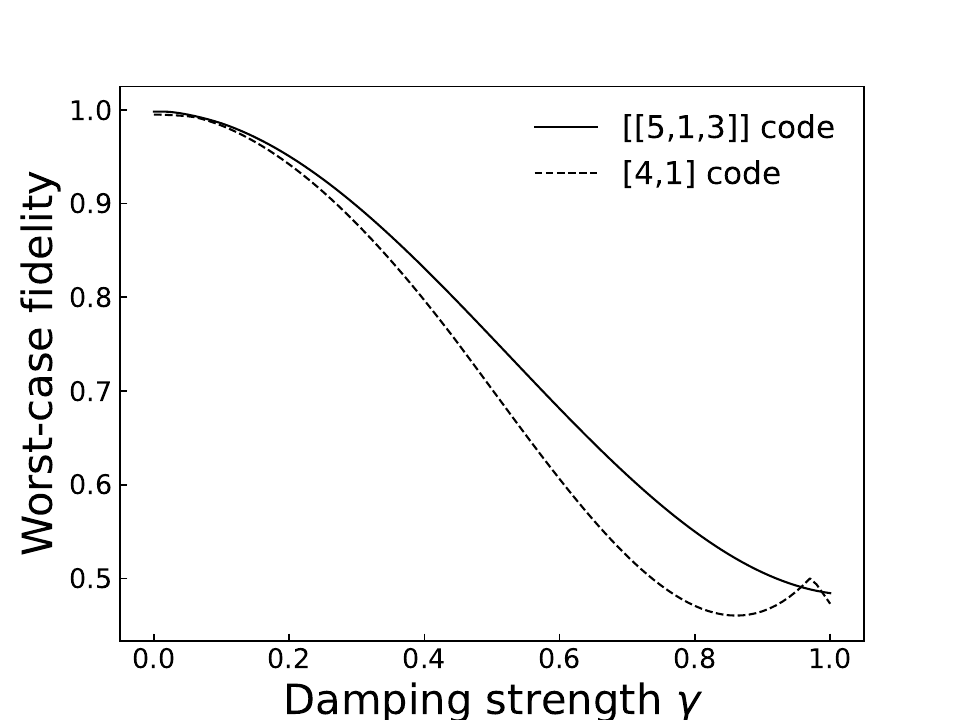}
\caption{Behavior of the worst-case fidelity under the non-unital combined channel $\mathcal{R} \circ \mathcal{E}$ for [[5,1,3]] (the Bold curve) and the [[4,1]] (the dotted curve) codes. Here, the Petz map $\cR$ is adapted to an amplitude-damping process of strength $\gamma=0.1$, for both codes.}
\label{fig}
\end{figure}

{\section{Bounds on fidelity for the Leung recovery adapted to non-Markovian noise }}\label{sec:leung_reco_app}
We begin this section with a review of the pioneering work of Leung {\it et al}~\cite{leung}, which laid out a framework for achieving approximate QEC by adapting the recovery and encoding scheme to the noise. Specifically, consider a code $\cC$ with associated projector $P$, and a noise channel $\cE$ with associated Kraus operators $\{E_{i}\}$. Let $p_l$ and $p_l\lambda_l$ respectively be the maximum and the minimum eigenvalues of the operators $PE_l^{\dagger}E_lP$. It was shown in Ref.~\cite{leung} that The code $\cC$ approximately corrects for  the errors $\{E_i\}$ with worst-case fidelity $F^2_{\rm min} \geq 1- \cO(\epsilon)$, iff,
\begin{align}
    (p_l - \lambda_lp_l)&\leq \cO(\epsilon), \nonumber \\
    PU_l^{\dagger}U_k P &= \delta_{lk}P,\label{eq:ortho_nm_leung}
\end{align} 
where $\{U_k\}$ are the unitary operators obtained from the polar decomposition of the error operators on the codesapce $E_kP =U_k \sqrt{PE_k^{\dag}E_kP}$. {Furthermore, it was shown that the desired fidelity can be achieved via a CPTP recovery map   
$\cR_L$  
with Kraus operators $R_i = P U_i^{\dagger}$.}

In order to {construct a noise-adapted recovery for non-Markovian noise in a similar fashion, we consider two different forms of the Leung recovery $\cR_{L}$.} First we consider the {n\"aive approach}, where the recovery operators are constructed {as above,} using the unitaries obtained via polar decomposition of the noise operators. {In other words, the recovery is \emph{exactly adapted} to the non-Markovian noise and we refer to this scheme as the non-Markovian Leung recovery, in our discussions below.}

In the second scenario we {consider a Markovian Leung recovery by constructing} the recovery operators {using} unitary operators that are obtained {via} polar decomposition of the noise operators {whose noise strength is} adapted to the {Markovian} regime of the noise. {The two approaches obviously yield different bounds on the worst-case fidelity.} 

\subsection{Non-Markovian Leung recovery}\label{sec:leung-theorem}

\begin{lemma}
     Consider a HPTP map $\cE(t)\sim\{sign(l),A_l(t)\}$ and a code $\cC$ with projector $P$. Let $\lambda_lp_l$ and $p_l$ be the smallest and largest eigenvalues, respectively, of the operator $\sqrt{P A_l^{\dag}A_lP}$, where $\{A_l\} $ are {a set of errors such that, $PU_l^{\dagger}U_k P = \delta_{lk}P$, for unitaries $\{U_{l}\}$ defined as $E_{l}P = U_{l}P\sqrt{PE_{l}^{\dagger}E_{l}P}$. The worst-case fidelity is then bounded as follows.}
     \begin{align}
         F^2_{\rm min}&\geq \sum\limits_{l} p_l\lambda_l, 
     \end{align}
     and the upper bound on the $F^2_{\rm min}$ follows as 
     \begin{align}\label{eq:fid_max_nm_leung}
         F^2_{\rm min}&\leq  1- \sum\limits_{l}{\rm sign}(l)\left(\langle\psi|\pi_l^{\dagger}\pi_l|\psi \rangle- |\langle\psi|\pi_l|\psi \rangle|^2\right)\nonumber\\
         & - \sum\limits_{kl}\left(\langle\psi|D_{kl}^{\dagger}D_{kl}|\psi \rangle- |\langle\psi|D_{kl}|\psi \rangle|^2\right)\nonumber\\
    &- \sum\limits_{kl} p_l\lambda_{l}\left(\langle\psi|M_{kl}^{\dagger}M_{kl}|\psi \rangle- |\langle\psi|M_{kl}|\psi \rangle|^2\right),
     \end{align}
     where $M_{kl}= PU_k^{\dagger}N_lP$, $\{N_l\}$ is the set of error operators which the not correctable by the code $\cC$, and $D_{kl}=M_{kl}\pi_l$. Here the operator $\pi_l$ is a residue operator which has the following form \cite{leung} 
     \begin{align}
         \pi_l = \sqrt{PE_l^{\dagger}E_lP} - \sqrt{\lambda_l p_l}P.
     \end{align}
\end{lemma}
\begin{proof}
  To appear in the upper bound of the fidelity we consider the recovered state first. The recovered state is the following 
  \begin{align}\label{eq:rec_leung_nm}
      \cR_L \circ \cE[|\psi\rangle\langle\psi|]&= \sum\limits_{kl}{\rm sign}(k){\rm sign}(l)R_k E_l |\psi\rangle\langle\psi| E_l^{\dagger}R_k^{\dagger} \nonumber\\
      &+\sum\limits_{kl}{\rm sign}(k){\rm sign}(l)R_k N_l |\psi\rangle\langle\psi| N_l^{\dagger}R_k^{\dagger}\nonumber\\
      & +\sum\limits_{l}{\rm sign}(l)P_EA_k |\psi\rangle\langle\psi| A^{\dagger}_k P_E.
  \end{align}
   
  Now considering the orthogonality condition in Eq.\eqref{eq:ortho_nm_leung} along with the following completeness condition 
  \begin{align}
      \sum\limits_{lk} E_l^{\dagger}R_k^{\dagger}R_k E_l + \sum\limits_{l}N_l^{\dagger} P_EN_L = I,
  \end{align}
  and the triangle inequality we obtain the bound in Eq.\eqref{eq:fid_max_nm_leung}. The lower bound follows from the Eq.\eqref{eq:rec_leung_nm} as
  \begin{align}
      F^2_{\rm min} \geq \sum\limits_{kl}{\rm sign}(k){\rm sign}(l)|\langle\psi|P U_k^{\dagger} U_lP (\sqrt{\lambda_lp_l}+\pi_l ) |\psi\rangle|^2.
  \end{align}
  Now due to the orthogonality condition in Eq.\eqref{eq:ortho_nm_leung}, we obtain
  \begin{align}
       F^2_{\rm min} \geq \sum\limits_{l}\lambda_lp_l ({\rm sign}(l))^2 = \sum\limits_{l}\lambda_lp_l.
  \end{align}
\end{proof} 

Here, we have considered the case of Leung recovery adapted to the non-Markovian regime of the noise, where the lower bound on the fidelity does not contain the ${\rm sign}(l)$ factor, while the upper bound does. On the other hand, the Leung recovery adapted to the Markovian regime of the noise yields fidelity bounds in which the sign$(l)$ factor that does not vanish, whose details are derived in {the following subsection}. Again, this is expected since for the non-Markovian noise, the uncorrectable part contributes more within a short span of time, causing the fidelity to fall rapidly. This fact becomes apparent in the example we study {in Sec.~\ref{sec:examples}}.

\subsection{Markovian Leung recovery}\label{sec:markov_leung}

\begin{theorem}\label{lem:leung_lem}(Fidelity bounds for Markovian Leung recovery). Consider a HPTP map $\cE(t)\sim\{sign(l),A_l(t)\}$, and a $d$-dimensional code $\cC$ with projector $P$. Let $\lambda_lp_l$ and $p_l$ be the smallest and largest eigenvalues, respectively, of the operator $\sqrt{P E_i^{\dag}E_iP}$, where $\{E_l\} $ are the set of correctable errors. A code $\cC$ is said to be correctable under a recovery operation $\cR_L(.) \sim \sum\limits_{k} PU_k^{m}(.)U_{k}^{m\dagger}P$, which is adapted to the Markovian regime of the noise process, such that the following conditions are satisfied 
\begin{align}\label{eq:ortho_markov_leung}
    PU^{m\dagger}_kU_lP &= \delta_{kl} (\beta_{l}P+\Delta_{ll}),
\end{align}
where $U_k^{m}$ are the unitary operators obtained form the polar decomposition of $E_k^m P$, and $E_k^m$ are the Kraus operators of the noise in the Markovian regime.

 The lower bound of the worst-case fidelity takes the following form 
\begin{align}
     F^2_{\rm min} &\geq \underset{\ket{\psi}}{\rm min}  \sum\limits_{l} {\rm sign}(l) \lambda_lp_l \beta_{l}^2,
\end{align}
where $\beta_{l}$ is real number and $\Delta_{ll}$ is a diagonal operator in $\cB(\cC)$. And the upper bound for the worst-case fidelity is the following 
\begin{align}\label{eq:max_wcf_leung}
    F^2_{\rm min} &\leq 1- \sum\limits_{l}{\rm sign}(l)\beta^2_{l}\left(\langle\psi|\pi_l^{2}|\psi \rangle- |\langle\psi|\pi_l|\psi \rangle|^2\right)\nonumber\\ &-\sum\limits_{k,l}{\rm sign}(l)\left(\langle\psi|\Delta_{ll}^{2}|\psi \rangle- |\langle\psi|\Delta_{ll}|\psi \rangle|^2\right)\nonumber\\
    & - \sum\limits_{k,l} p_l\lambda_{l}{\rm sign}(l)\left(\langle\psi|\pi_l \Delta_{ll}^{2}\pi_{l}|\psi \rangle- |\langle\psi|\pi_l\Delta_{ll}|\psi \rangle|^2\right) -\nonumber\\
    &-\sum\limits_{k,l} p_|\lambda_{l}{\rm sign}(l)\left(\langle\psi| M_{kl}^{m\dagger}M^m_{kl}|\psi \rangle- |\langle\psi|M^m_{kl}|\psi \rangle|^2\right),
\end{align}
\end{theorem}

\begin{proof} Consider $\{E_k^m\}$ are the set of Kraus operators for the noise in the Markovian limit, which satisfies the following conditions
\begin{align}
    p_k^m(1-\lambda_k^m)&\leq\cO(\epsilon^{t+1}),\\
     PU_l^{m\dagger}U^m_k P &= \delta_{lk}P \label{eq:uni_leung},
\end{align}
where $p_k^m$ and $p_k^m\lambda_k^m$ are the largest and the smallest eigenvalues, respectively, of the operator $\sqrt{PE_k^{m\dagger}E^m_k P}$, and the unitary-operators $U_k$'s are the unitary from the polar decomposition of the operator $E^m_kP $,
\begin{align}
    E^m_kP=U^m_k \sqrt{PE_k^{m\dagger}E^m_k P}.
\end{align}
Since the unitaries $U^m_k$'s satisfy the conditions in Eq.\eqref{eq:uni_leung}, we can  construct a CPTP recovery map $\cR_L = \sum\limits_{k} R_k [.]R_k^{\dagger}+P_E [.]P_E$, where the operators $R_k = PU_k^{m\dag}$, and $P_E=I-\sum\limits_K R^{\dagger}_k R_k$. 

The worst-case fidelity $F^2_{\rm min}$ between the initial state $\ket{\psi} \in \cC$ and the state $\cR_L\circ\cE[|\psi\rangle\langle\psi|]$ is 
\begin{align}\label{eq:fmin_markov_leung}
    F^2_{\rm min}&= \underset{\ket{\psi}}{\rm min} \sum\limits_{k,l} {\rm sign}(l)|\bra{\psi}R_k A_l\ket{\psi}|^2 + |\bra{\psi}P_E \ket{\psi}|^2\\
    &\geq \underset{\ket{\psi}}{\rm min}  \sum\limits_{k,l} {\rm sign}(l)|\bra{\psi}R_k E_l\ket{\psi}|^2  \end{align}
Now, let $U_l$ be a unitary operator obtained from the polar decomposition of the operator $E_lP = U_l\sqrt{PE_l^{\dagger}E_l P }$, where $E_l$ are the HPTP error operators. From triangle inequality, we have
   \begin{align}
       F^2_{\rm min}&\geq \underset{\ket{\psi}}{\rm min}  \sum\limits_{k,l} {\rm sign}(l)|\bra{\psi}PU_k^{m\dagger}U_l P \sqrt{PE_k^{\dagger}E_k P}\ket{\psi}|^2.
\end{align}
Now, consider the following decomposition 
\begin{align}
    \sqrt{PE_k^{\dagger}E_k P }&= \sqrt{\lambda_kp_k} P+ \pi_k.
\end{align} The operator norm of the residue operator $\pi_k$ has the following bound
\begin{align}
    0\leq||\pi|| \leq \sqrt{PE_k^{\dagger}E_k P}- \sqrt{\lambda_kp_k}P.
\end{align}
Under these constraints along with following conditions
\begin{align}
    PU^{m\dagger}_kU_l P&= \delta_{kl}(\beta_{k} P+\Delta_{kk}),
\end{align}
where $\beta_{kk}$ is diagonal on the support of the codespace, the lower bound of the fidelity takes the following form 
\begin{align}
     F^2_{\rm min} &\geq \underset{\ket{\psi}}{\rm min}  \sum\limits_{k,l} {\rm sign}(l)|\bra{\psi} \beta_{l}P ( \sqrt{\lambda_lp_l} P+ \pi_l)\ket{\psi}|^2 \\
        F^2_{\rm min} &\geq \underset{\ket{\psi}}{\rm min}  \sum\limits_{l} {\rm sign}(l) \lambda_lp_l\beta_{l}^2.
\end{align}

For the upper bound we consider the Eq.\eqref{eq:fmin_markov_leung} and the completeness following relation
\begin{align}\label{eq:comp_markov_leung}
    \sum\limits_{kl} {\rm sign}(l)E_l^{\dagger}R_k^{\dagger}R_k E_l +  \sum\limits_{kl} {\rm sign}(l)N_l^{\dagger}R_k^{\dagger}R_k N_l\nonumber \\
    + \sum\limits_{k}A_k^{\dagger}P_EA_k= I .
\end{align}
Now using the orthogonality relation  in Eq.\eqref{eq:ortho_markov_leung}, and the Eq.\eqref{eq:fmin_markov_leung} and the Eq.\eqref{eq:comp_markov_leung} we obtain upper bound of $F^2_{\rm min}$ as in Eq.\eqref{eq:max_wcf_leung}.  

\end{proof}
Here we note one crucial difference between the non-Markovian Leung recovery and the Markovian one. In the Markovian Leung recovery the lower bound of the fidelity contains a ${\rm sign}(i)$ factor, where as the non-Markovian Leung does not contain this factor in the lower bound of $F^2_{\rm min}$. Along with this difference the in exact noise adaptation causes some new terms to appear in the fidelity expression. These terms cause the fidelity under Markovian recovery to fall faster than the non-Markovian one.   

As a consequence of this, we observe  
Hence, as long as the $F^2_{\rm min} \ge 1- \cO(\epsilon^{t+1})$, the code $\cC$ correct the effect of the noise $\cE$ up to $t-$order of its noise strength.

 \begin{figure}
      \centering
      \includegraphics[width=1.0\linewidth]{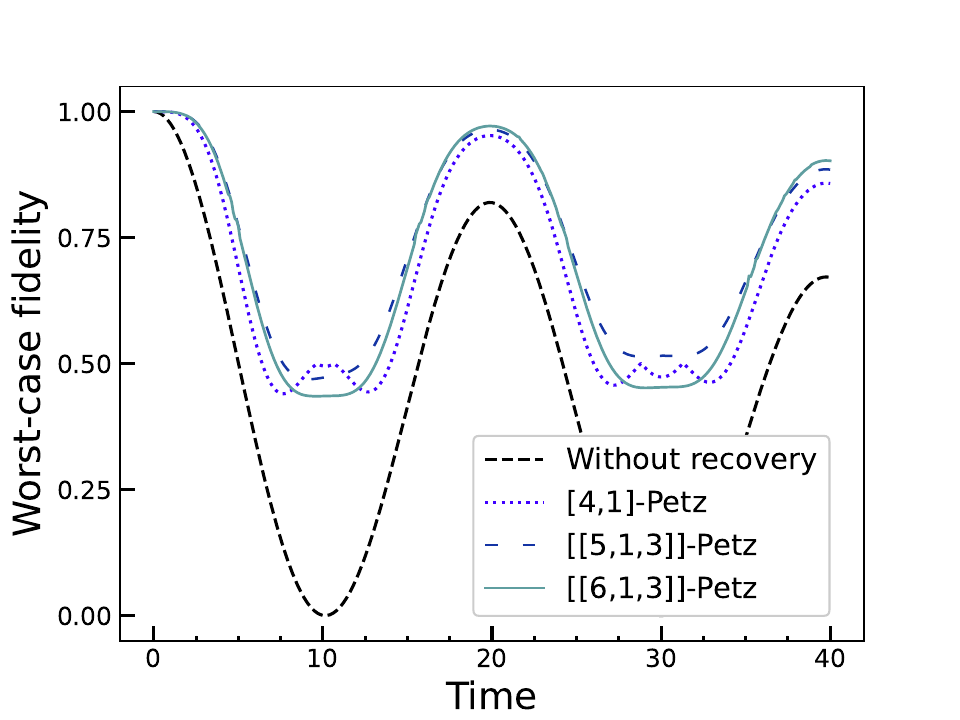}
      \caption{Markovian Petz recovery for different codes. Notice that the fidelity for [[5,1,3]] is better compared to other codes. And in the long time limit, it preserves the code space.}
      \label{fig:M_petz_diff}
  \end{figure}

\section{Correction of non-Markovian AD with the $[[5,1,3]]$ code with Markovian Petz recovery \label{sec:5-6-code-remedy}}
Consider the $[[5,1,3]]$ code. For non-Markovian amplitude damping channel (see Appendix~\ref{sec:NMAD}), we can see the Knill-Laflamme conditions do not hold for the noise operator $E_{00000}=E_0^{\otimes 5}$ as the value of the damping strength $\gamma$ approaches the maximum value: 
{\small  \begin{align}
 \langle 0_L|E^{\dagger}_{00000}E_{00000}|0_L\rangle &=   1 -\frac{5 \gamma}{2}+\frac{5 \gamma^2}{2} -\frac{5 \gamma^3}{4} + \frac{5 \gamma^4}{16}\\
 \langle 1_L|E^{\dagger}_{00000}E_{00000}|1_L\rangle &=1 -\frac{5 \gamma}{2} + \frac{5 \gamma^2}{2} - \frac{5 \gamma^3}{4} + \frac{5 \gamma^4}{16} -\frac{\gamma^5}{16},
\end{align}}where $\gamma$ is the damping strength. In this sense, the $[[5,1,3]]$ code is not a perfect code for amplitude damping noise and the Petz channel does not reduce to syndrome extraction followed by unitary recovery. When we choose Petz channel as near-optimal recovery, it again outperforms the standard syndrome-based recovery as well as Leung recovery, specifically when the chosen stabilizer code is not perfect, here due to large noise strength, see Figure (\ref{fig:M_petz_diff}). In fact, one may notice the advantage [[5,1,3]] has when the $\cR \circ \cE$ channel is non-unital, which is demonstrated in Appendix \ref{sec:inexact-demo}.

\end{document}